\documentclass{article}

\usepackage{PRIMEarxiv}

\usepackage[utf8]{inputenc} 
\usepackage[T1]{fontenc}    
\usepackage{hyperref}       
\usepackage{url}            
\usepackage{booktabs}       
\usepackage{amsfonts}       
\usepackage{nicefrac}       
\usepackage{microtype}      
\usepackage{lipsum}
\usepackage{fancyhdr}       
\usepackage{graphicx}       
\usepackage{natbib}
\usepackage{amsmath}
\usepackage{amssymb}
\usepackage{mathtools}
\usepackage{amsthm}
\usepackage{algorithm}
\usepackage{algorithmic}
\usepackage{subcaption}
\usepackage{natbib}

\graphicspath{{media/}}     

\newtheorem{theorem}{Theorem}
\newtheorem{proposition}{Proposition}
\title{Model-X Change-Point Detection of Conditional Distribution
}

\author{%
  Zhuofan Dong\thanks{Equal Contribution} \\
  School of Mathematical Sciences \\
  Beijing Normal University \\
  \texttt{202211999077@mail.bnu.edu.cn} \\
  \And
  Yiwen Huang\footnotemark[1]  \\
  Department of Statistics \\
  Peking University \\
  \texttt{2401210080@stu.pku.edu.cn} \\
  \And
  Yan Dong\footnotemark[1] \\
  International Institute of Finance, School of Management \\
  University of Science and Technology of China \\
  \texttt{apriloa@mail.ustc.edu.cn} \\
  \And
  Mengying Yan \\
  Department of Biostatistics \& Bioinformatics \\
  Duke University School of Medicine \\
  \texttt{mengying.yan@duke.edu} \\
  \And
  Ziye Tian \\
  Department of Biostatistics \& Bioinformatics \\
  Duke University School of Medicine \\
  \texttt{ziye.tian@duke.edu} \\
  \And
   Chuan Hong\\
   Department of Biostatistics \& Bioinformatics\\
   Duke University School of Medicine\\
   \texttt{chuan.hong@duke.edu} \\
  \And
  Doudou Zhou\thanks{Corresponding author. To whom correspondence should be addressed.} \\
  Department of Statistics and Data Science \\
  National University of Singapore \\
  \texttt{ddzhou@nus.edu.sg} \\
  \And
  Molei Liu \footnotemark[2]\\
  Department of Biostatistics, Peking University Health Science Center \\
  Beijing International Center for Mathematical Research, Peking University \\
  \texttt{moleiliu@bjmu.edu.cn} \\
}

\begin{document}
\maketitle

\begin{abstract}
The dynamic nature of many real-world systems can lead to temporal outcome model shifts, causing a deterioration in model accuracy and reliability over time. This requires change-point detection on the outcome models to guide model retraining and adjustments. However, inferring the change point of conditional models is more prone to loss of validity or power than classic detection problems for marginal distributions. This is due to both the temporal covariate shift and the complexity of the outcome model. Also, the existing method of conditional change points detection both have many limitations including linear assumption and low dimension prerequisite which sometimes is not suitable for real world application. To address these challenges, we propose a novel \textbf{M}odel-X chang\textbf{E}-point detectio\textbf{N} of conditional \textbf{D}istribution (MEND)  method computationally enhanced with distillation function for simultaneous change-point detection and localization of the conditional outcome model. We extend and combine our model with neural network to accommodate complex nonlinear and high dimensional situation, which is proved to be valid in both simulation and real data.  Theoretical validity of the proposed method is justified. Extensive simulation studies and two real-world examples demonstrate the statistical effectiveness and computational scalability of our method as well as its significant improvements over existing methods.
\end{abstract}


 \section{Introduction}
\label{sec:intro}

\subsection{Background}

The advancement of experimental technologies has made large-scale data featuring high-dimensional covariates readily available. Modern developed machine learning techniques have leveraged such big data to develop prediction models for uncovering underlying patterns and facilitating data-driven decision-making. However, due to the dynamic nature of many real-world systems, model performances may shift over time, leading to inaccuracies in models \citep{davis2020detection}. For example, a case study on Acute Kidney Injury (AKI) found that seven parallel machine learning models, trained using data from 2003, showed deteriorating model performance when evaluated over the next nine years in a validation cohort \citep{davis2017calibration}.  Therefore, it becomes necessary to develop a change point detection method for detecting shifts in model performance, which could indicate that the current model is no longer effective.  

In recent years, there has been a growing literature on the change point detection methods for identifying the possible changes in the distribution properties of temporal data sequences. However, most existing change point detection methods predominantly focus on identifying changes in the marginal distribution of the observations. These methods are not applicable for predictive modeling tasks, of which the primary interest lies in understanding whether and when the relationship between a response and covariates. For example, changes in the conditional distribution of health outcomes given clinical indicators \citep{gao2025causal}, or stock returns conditional on market factors \cite{vogt2015}, may significantly affect the reliability of predictive models over time. Similar issues are also encountered in environmental analysis \citep{Cai2000}, power systems \citep{Chen2016}, and meteorology \citep{Harris2022}. 


Furthermore, many existing change-point detection methods suffer from certain limitations, making them less adaptable to real-world applications. To ensure statistical validity, studying changes in dependency structures often requires a large number of samples, while the data themselves are typically high-dimensional. For example, in gene expression analysis, financial asset co-movement modeling, and multi-modal sensor signal monitoring, the dimensionality of the data is often extremely high, and the variables exhibit complex nonlinear or higher-order dependencies. In such cases, traditional methods that rely on linear assumptions or low-dimensional approximations tend to fail in robustly identifying the true change-point structure.

At present, the majority of existing change-point detection methods are constructed under explicit modeling assumptions; for example, the OLS-CUSUM \citep{ploberger1992cusum} test is based on a linear regression model. We want to develop a framework that does not rely on any assumptions about the underlying data distribution, thus enabling broad applicability to a wide range of both linear and nonlinear complex cases.

Therefore, our goal is to develop a method for detecting and localizing change points in conditional distributions under high-dimensional, complex settings with limited labeled data.

\subsection{Our contribution}

First, In this paper, we propose the \textbf{M}odel-X chang\textbf{E}-point detectio\textbf{N} of conditional \textbf{D}istribution (MEND). Our model leverages both labeled and unlabeled data to comprehensively extract informative patterns and to enhance statistical validity under limited supervision. Inheriting the attractive properties of the model-X framework, the proposed approach does not impose any specific structural or distributional assumptions on the conditional distribution of responses given covariates, thereby ensuring robustness across diverse linear and nonlinear scenarios.

Second, to address high-dimensional and complex data settings, we introduce a neural-network-based functional encoder as a distilled representation module. Distinct from the explicit and fixed functional constructions adopted in dCRT, our encoder is data-adaptive: its layer depth, activation functions, and architecture are tailored to the data characteristics, pretrained, and then frozen as a fixed distillation function. Unlike the dCRT’s distillation that mainly aims to reduce computational cost, our encoder’s primary purpose is to map covariates into a time-invariant latent subspace, where the intrinsic structure of the covariates remains stable over time. When encountering change point, the discrepancy between these stable latent representations and the evolving responses becomes amplified, thereby substantially enhancing the model’s sensitivity to distributional shifts.

Finally, for scenarios with at most a single change point, we propose a simplified version of our framework that eliminates unnecessary network structures and complex test statistics while maintaining high detection sensitivity. This lightweight version offers a more efficient and interpretable alternative, achieving comparable detection accuracy with dramatically reduced computational cost when the number of change points is known to be at most one.

\section{Problem setup and notation}

Unlike traditional time-series problems where each time point corresponds to a single observation, our study focuses on identifying the time points of dependency changes in a setting where multiple independent and identically distributed (i.i.d.) samples are observed at each time point, and the number of samples may vary across time. Assume we observe data over a total time horizon of $T$. At each time point $t \in \{1, \ldots, T\}$, we independently sample $n_t$ observations, forming a dataset 
\[
D_t = \{ (Y_i^{[t]}, \mathbf{X}_i^{[t]}) : i = 1, \ldots, n_t \},
\]
where $Y_i^{[t]} \in \mathbb{R}$ denotes the response variable and $\mathbf{X}_i^{[t]} \in \mathbb{R}^p$ represents the corresponding covariates.  
We denote by $D = \{ (Y_i, \mathbf{X}_i, R_i) \}$ the overall collection of samples, where $R_i$ indicates the time index associated with the $i$-th sample, i.e., $(Y_i, \mathbf{X}_i) \in D_t$ implies $R_i = t$.
Our goal is to test whether the conditional distribution of $Y$ given $\mathbf{X}$ remains invariant across different time points. Formally, the null hypothesis is
\begin{equation}
    H_0 : \; p_{Y|\mathbf{X},R}(y) = p_{Y|\mathbf{X}}(y), \quad \forall R \in \{1, \ldots, T\}, \label{H0}
\end{equation}

that is, the dependency structure between $Y$ and $\mathbf{X}$ does not change over time. In other words, it means that  \( Y \) is conditionally independent of \( R \) given \( \mathbf{X} \), denoted as \( Y \perp\!\!\!\perp R \mid \mathbf{X} \). 
The corresponding alternative hypothesis is
\begin{equation}
 \begin{aligned}
    H_1 : \; \exists \, 1 \le \tau < T, \; \text{s.t.} \;
p_{Y|\mathbf{X},R}(y) =
\begin{cases}
p_{Y|\mathbf{X}}(y), & R = \tau-1, \\
q_{Y|\mathbf{X}}(y), & R = \tau,
\end{cases} 
 \end{aligned}
  \label{H1}
\end{equation}
where $p_{Y|\mathbf{X}} \neq q_{Y|\mathbf{X}}.$ We can see that the proposed assumptions are inherently broad and encompass both the multiple-change-point and the single-change-point settings. More specifically, the alternative hypothesis $H_1$ is simply the complement of the null hypothesis $H_0$, which naturally accommodates scenarios with either a single change point or multiple change points.
In conclusion, we aim to identify potential change points $\tau$ where the conditional relationship between $Y$ and $\mathbf{X}$ undergoes a distributional shift.

\section{Method}
\label{sec:method}

\subsection{Overview of MEND}

We build our testing framework \textbf{MEND} upon the Conditional Randomization Test (CRT), originally introduced by \citet{Candes2018PanningFG}, belonging to the broader model-X framework. The CRT provides an exact (non-asymptotic) control of the Type-I error for conditional independence testing, under the assumption that the conditional distribution \( R \mid \mathbf{X} \) is known (or can be accurately estimated).

Operationally, CRT relies on a user-defined test statistic \( S \), which can be designed to capture task-specific structure. 
Given that the validity of CRT is invariant to the choice of \( S \), we are free to incorporate complex, possibly nonlinear functions such as neural-network–based encoders to distill informative representations from high-dimensional data. 
CRT then constructs counterfactual datasets by resampling \( R_i \) while preserving the conditional structure of \( R_i \mid \mathbf{X}_i \). 
For each resampled dataset, the test statistic \( S \) is recomputed, and the \( p \)-value is obtained by ranking the observed statistic among the resampled values. 
The procedure is summarized in Algorithm~\ref{alg:crt_modified}, and its formal validity is established in Theorem~\ref{theo:valid_test}, proof can be found in Appendix \ref{appendix:proof}.


\begin{algorithm}
\caption{\label{alg:crt_modified} Conditional Randomization Test (CRT) Framework}
\begin{algorithmic}[1]
\STATE \textbf{Input:} A set of \( n \) independent samples \(D= \{(Y_i, \mathbf{X_i},R_i)\}_{1 \leq i \leq n}\). Test statistic \( S(D) \), number of randomization \(K\), and nominal significance level \( \alpha \in (0,1) \).  

\STATE \textbf{Output:} A \( p \)-value for testing $ H_0$ against $H_1$ defined in (\ref{H0}) and (\ref{H1}).

\STATE \textbf{Randomization:}
\FOR{\( k = 1, 2, \dots, K ~\&~ i=1,\dots,n\)}
    \STATE Randomly generate \( R_i^{(k)} \) from \( p_{R|\mathbf{X_i} }(t) \) independently from $(Y_i,R_i)$, while keeping \( \mathbf{X_i} \) and \( Y_i \) unchanged, resulting in the resampled data set \(D^{(k)} = \{(Y_i, \mathbf{X_i},R^{(k)}_i )\}_{1 \leq i \leq n}\). 
\ENDFOR

\STATE \textbf{Compute $p$-value:}
\[
p\text{-value} = \frac{1}{K+1} \left( 1 + \sum_{k=1}^K \mathbb{I}(S(D^{(k)}) \geq S(D)) \right),
\]
where \(\mathbb{I}( \cdot )\) is the indicator function. Reject $H_0$ if  \( p \)-value $\leq \alpha$.  
\end{algorithmic}
\end{algorithm}

\begin{theorem}[Strict validity of CRT]
\label{theo:valid_test}
When \( H_0: p_{Y|\mathbf{X},R}(y) = p_{Y|\mathbf{X}}(y) \) holds for \( R = 1, \ldots, T \),
\begin{align*}
    \lim_{n \to \infty} \mathbb{P}(\textnormal{Algorithm \ref{alg:crt_modified} rejects}) &\leq \alpha.
\end{align*}
\end{theorem}

In constructing the test statistic, the covariates $\mathbf{X}$ are often high-dimensional.
Directly performing the Conditional Randomization Test (CRT) resampling and computing the statistic $S(D^{(k)})$ in a high-dimensional space 
would result in heavy computational cost and may also lead to instability.
To address this, we follow the idea of the distilled CRT (dCRT) proposed by \citet{liu2022fast},
and introduce a distillation function $\phi$ to achieve dimensionality reduction and information extraction.
Without using the time label $R$, the function $\phi$ extracts predictive information in $\mathbf{X}$ that is highly correlated with $Y$,
so that the original statistic $S(D) = S(Y, \mathbf{X}, R)$ can be simplified to 
$T(Y, \mathbf{d}, R) = T(Y, \phi(\mathbf{X}), R)$,
which significantly reduces the computational burden during CRT resampling.

In contrast to the original dCRT framework, which requires manually constructing explicit distillation functions, the design of such functions heavily relies on prior knowledge of the data structure and is often difficult to generalize to complex settings. Moreover, the classical dCRT is primarily motivated by alleviating the computational burden induced by high-dimensional data.
In our model, the proposed encoder $\phi$ is implemented as a multi-layer neural network that requires no distributional or functional assumptions. Through pretraining and subsequent freezing, $\phi$ defines an implicit mapping without any closed-form expression. A key purpose of this encoder is to project $\mathbf{X}$ into a low-dimensional latent space with improved structural stability, such that discrepancies between $\mathbf{X}$ and $Y$ before and after a change point are effectively amplified, thereby enhancing the sensitivity of the detection
procedure.

More precisely, suppose the number of change points is $k-1$, so that the entire time series is
partitioned into $k$ segments. We associate each segment with a parameter vector
$\{\omega_1,\ldots,\omega_k\}$. If an observation $(\mathbf{x}_j, y_j)$ belongs to the $i$-th
segment, its expected response is modeled as
\[
    y_j = \omega_i^\top\phi(\mathbf{x}_j) + \epsilon_j,
\]
where $\epsilon_j$ denotes random noise.

From an encoder--decoder perspective, the proposed stable encoder $\phi$ maps $\mathbf{X}$ into
a low-dimensional and structurally invariant latent space, while different linear decoders
$\omega_1,\ldots,\omega_k$ capture segment-specific structures and generate predictions for $Y$.
This multi-decoder design enables the model to flexibly characterize subtle structural differences
across segments while sharing a globally stable encoder, thus improving change-point
sensitivity without sacrificing representation stability.
 Such pipeline is summarized in Algorithm~\ref{alg:crt_dist}.

The depth and activation functions of the multilayer neural network $\phi$ 
are chosen based on the characteristics of the data.
Since our study focuses on multivariate time-series data,
we adopt a four-layer fully connected neural network with ReLU activations applied to the first two layers to alleviate gradient vanishing.
For other data modalities, the network design can vary:
CNNs may be used for image data to capture spatial structure,
while GNNs are suitable for graph-structured or highly correlated features.

\begin{algorithm}[t]
\caption{\label{alg:crt_dist} Distilled Conditional Randomization Test (CRT) Framework}
\begin{algorithmic}[1]
\STATE \textbf{Input:} A set of \( n \) independent samples \( D= \{(Y_i, \mathbf{X_i} ,R_i)\}_{1 \leq i \leq n}\). Test statistic \( S(D) \), number of randomization \(K\), and nominal significance level \( \alpha \in (0,1) \). Neural network \(\phi\)  

\STATE \textbf{Output:} A \( p \)-value for testing $ H_0$ against $H_1$ defined in (\ref{H0}) and (\ref{H1}).

\vspace{0.8ex}
\textbf{Stage 1: Pretraining the encoder}\\
\vspace{0.8ex}
\STATE Initiate distillation network: \(\phi\)  

\STATE Train encoder $\phi$ on $\{(\mathbf{x}_i, y_i)\}_{i=1}^n$ by iterating the EM algorithm (\ref{app:alg1})  until converged;  \\

\vspace{0.8ex}
\textbf{Stage 2: Conduct CRT}\\
\vspace{0.8ex}

\STATE \textbf{Randomization:}
\FOR{\( k = 1, 2, \dots, K ~\&~ i=1,\dots,n\)}
    \STATE Randomly generate \( R_i^{(k)} \) from \( p_{R|\mathbf{X_i} }(t) \) independently from $(Y_i,R_i)$, while keeping \( \mathbf{X_i} \) and \( Y_i \) unchanged, resulting in the resampled data set \(D'^{(k)} = \{(Y_i, \phi(\mathbf{X_i}),R^{(k)}_i )\}_{1 \leq i \leq n}\). 
\ENDFOR

\STATE \textbf{Compute $p$-value:}
\[
p\text{-value} = \frac{1}{K+1} \left( 1 + \sum_{k=1}^K \mathbb{I}(S(D'^{(k)}) \geq S(D')) \right),
\]
where \(\mathbb{I}( \cdot )\) is the indicator function. Reject $H_0$ if  \( p \)-value $\leq \alpha$.  
\end{algorithmic}
\end{algorithm}

\subsection{Training the network}

Assuming $K-1$ change points, we introduce a latent variable $E_i \in \{1, \ldots, K\}$ 
for each sample $(\mathbf{x}_i, y_i)$ to represent the temporal segment to which it belongs.
We define the responsibility $p_{ik}$ as the posterior probability that sample $i$ belongs to segment $k$.
Hence, the overall expected output can be expressed as: 
\[
\hat{y_i } = \sum_{k=1}^{K} p_{ik}\, \omega_k^\top \phi(\mathbf{x}_i).
\]
Conditioned on $E_i = k$, the model simplifies to a local linear decoder:
\[
y_i \mid E_i = k = \omega_k^\top \phi(\mathbf{x}_i) + \varepsilon_k, \quad \varepsilon_k \sim \mathcal{N}(0, \sigma_k^2),
\]
or equivalently,
\[
y_i - \omega_k^\top \phi(\mathbf{x}_i) \mid E_i = k \sim \mathcal{N}(0, \sigma_k^2).
\]
In other words, the model can be regarded as a Latent Mixture Model (LMM),
where $\phi$ acts as a shared encoder and $\omega_k$ represents local linear decoders for different temporal segments.

Under this assumption, the encoder $\phi$ is pretrained. More specifically, we adopt the classical EM algorithm to perform iterative optimization. In the E-step, we compute the posterior responsibilities for each sample based on the current parameters. In the M-step, we fix these responsibilities and update $\phi$ and $\{\omega_i\}$ via gradient-based optimization, while the noise variances $\{\sigma_i^2\}$ admit closed-form updates. The detailed procedures are provided
in the Appendix \ref{app:alg1}.

The E-step and M-step are iterated until convergence (e.g., when parameter changes fall below a threshold).
In scenarios with strong nonlinearity or severe temporal drift, 
regularization methods such as Invariant Risk Minimization (IRM) can be incorporated 
to further encourage $\phi$ to learn stable, time-invariant representations (see Appendix \ref{app:alg1} for details).

Our framework enjoys two complementary forms of validity. First, the latent mixture
structure ensures that change-point estimation remains valid even when the encoder
$\phi$ is misspecified. The EM updates depend only on the mixture likelihood and the
segment-wise decoders, so the latent segmentation is identifiable without requiring the representation model to be correct. Second, if a CRT-style randomization test is employed for inference, its model-agnostic construction guarantees valid $p$-values irrespective of representation error. Together, these properties provide robustness at both the estimation level (recovering the correct segmentation) and the inference level (controlling Type-I error), thereby ensuring overall validity of the proposed framework.

\subsection{Construction of the Test Statistic $S(D)$}

After obtaining the trained distillation network $\phi$, 
we freeze its parameters and regard it as a fixed mapping, referred to as the \emph{distillation function}.
For each sample, we define
\[
\mathbf{Z}_i = \phi(\mathbf{X}_i) \in \mathbb{R}^p, 
\qquad 
\{\mathbf{Z}_t, Y_t\} = \{(\mathbf{Z}_i, Y_i) \mid R_i = t\},
\]
where $\mathbf{Z}_t$ represents the set of latent representations at time $t$.

During training, recall that each decoder is a single-layer linear mapping.
When training is sufficient, the nonlinearity between $\mathbf{X}$ and $Y$ has been fully extracted into $\phi$, 
so that the latent representation $\mathbf{Z}$ and the response $Y$ can be well approximated by a linear regression model.

To construct the test statistic $S(D)=S(Y,R,\mathbf{Z})$ for detecting change points, we make the following assumptions:
\begin{enumerate}
    \item Within the same time segment $t$, all samples $(\mathbf{Z}_i, Y_i)$ share a common regression coefficient $\beta_t$.
    \item If $t$ corresponds to a change point, the regression coefficients before and after $t$ differ significantly.
\end{enumerate}

Since multiple samples are available at each time point, 
after obtaining $\phi$, we estimate time-varying regression coefficients $\{\beta_t\}_{t=1}^T$
by solving a fused-regularized loss:

\begin{equation}
 \begin{aligned}
        \min_{\{\beta_t\}_{t=1}^T}
\sum_{t=1}^{T}&\frac{1}{2n_t}\|\mathbf{y}_t - \mathbf{Z}_t \beta_t\|_2^2
+ \lambda_{\mathrm{fuse}} \sum_{t=2}^{T}\|\beta_t - \beta_{t-1}\|_2
\\
&+ \frac{\lambda_{\mathrm{ridge}}}{2}\sum_{t=1}^{T}\|\beta_t\|_2^2
+ \lambda_{1}\sum_{t=1}^{T}\|\beta_t\|_1.
 \end{aligned}
\end{equation}

Here, $\lambda_{\mathrm{fuse}}$ encourages piecewise constant coefficients over time, 
while $\lambda_{\mathrm{ridge}}$ and $\lambda_1$ provide additional $\ell_2$ and $\ell_1$ regularization.
At each time $t$, let $\beta_t$ be the estimated regression coefficient so that 
\[
Y_t = \beta_t^\top \mathbf{Z}_t + \varepsilon_t.
\]

In periods where the underlying dependency structure remains stable, the temporal variation in the regression coefficients $\beta_t$ should reflect only random noise, and hence adjacent coefficients are expected to differ only mildly. In contrast, when a structural change occurs at time $t$, this stability is disrupted, and the coefficients before and after the change exhibit a pronounced deviation. Leveraging this property, we use the magnitude of the temporal jump in $\beta_t$ as the core diagnostic quantity and formulate the following criterion for detecting change points:
\[
D_t = \|\beta_t - \beta_{t-1}\|_2^2,\quad S(D) = \sum_{\text{top }(K-1)} D_t,
\]
that is, the sum of the largest $(K-1)$ squared coefficient jumps.
These corresponding time indices are selected as the estimated change points.
For more refined inference, one may further conduct hypothesis testing 
on individual $\|D_t\|_2^2$ values to obtain $p$-value based decisions.

\subsection{Single change point occasion}

Our proposed MEND framework can flexibly accommodate scenarios with multiple change points. However, when prior knowledge suggests that the data contain at most one single change point and the linear model is dominant, the model can be substantially simplified to improve computational efficiency.

In the single change-point setting, we simplify the neural network $\phi$ to a more parsimonious linear mixed model (LMM)--based statistical formulation that outputs the conditional means before and after the change point. Change-point detection is then carried out using a similar EM-type training procedure.

Formally, let $\tau$ denote an unknown change point. We define
\begin{equation}
\phi(X_i) = \{ m_0(X_i), \; m_1(X_i) \},
\end{equation}
where
\begin{equation}
\begin{aligned}
    m_0(X_i) = \mathbb{E}[Y_i \mid X_i,\; i < \tau], \\
m_1(X_i) = \mathbb{E}[Y_i \mid X_i,\; i \ge \tau].
\end{aligned}
\end{equation}
We refer to this simplified model as \textbf{MEND-mean}. Details of the training procedure are provided in Appendix \ref{app:Simplified}.

Building upon MEND-mean, we further consider a slightly more expressive variant that allows for nonlinear dependence on a single coordinate. This extension incorporates additional information while retaining computational tractability. Specifically, the distilled representation is defined as
\begin{equation}
\phi(X_i) = \{ m_0(X_i), \; m_1(X_i), \; X_B \},
\end{equation}
where $X_B$ denotes the selected coordinate. We refer to this model as \textbf{MEND-repr}. The corresponding training procedure is also described in Appendix \ref{app:Simplified}.

\section{Numerical Simulation}
We first simulate scenarios with complex dependency structures. To evaluate the performance of the proposed model under different distributional regimes, 
we construct a time-series dataset with potential change points and test the model 
in both linear and nonlinear, 
as well as single-change-point and multi-change-point scenarios.
We further estimate the predicted change-point locations 
and compare our method with several benchmark approaches.

We set the total time length to $T = 100$. 
At each time point $t$, we randomly sample the number of observations $n_t$ 
uniformly from the interval $[n_{\min}, n_{\max}]$, 
and construct the covariates as
$\mathbf{X} = [\mathbf{X}_e, \mathbf{X}_i],$
where
\begin{equation}
    \begin{aligned}
        \mathbf{X}_e - 0.5 \sim &\mathcal{N}(0, \Sigma) \in \mathbb{R}^{50}, 
    \Sigma_{ij} = 0.5^{|i-j|}, 
    \\
    & \mathbf{X}_i \sim \mathcal{N}(0, I) \in \mathbb{R}^{50}.
    \end{aligned}
\end{equation}

\paragraph{Linear Change-Point Scenario}
In the linear setting, we generate a coefficient vector $\alpha_k$ for each segment
and define the response variable as
$
Y_i = \alpha_k^\top \mathbf{X}_e + G(\mathbf{X}_e) + \epsilon_i,
$
where the shared nonlinear component is
$
G(\mathbf{X}) = 0.3 X_0^2 + \sqrt{|X_2|}.
$
To control the degree of coefficient variation across segments, we define
\[
\alpha_k = k_s \alpha_k' + (1 - k_s)\alpha_{\mathrm{all}}, \qquad 
\alpha_k', \alpha_{\mathrm{all}} \sim \mathcal{N}(0, I),
\]
where $k_s \in [0, 1]$ is a smoothness-strength parameter.
A smaller $k_s$ indicates that different segments share more similar coefficients 
(simulating mild changes), 
while a larger $k_s$ leads to more pronounced parameter shifts across change points.

\paragraph{Nonlinear Change-Point Scenario}

In the nonlinear setting, we introduce a family of nonlinear transformations
$F = \{f_1, f_2, f_3, f_4, f_5\}$ and a shared linear coefficient $\beta$.
For each segment, one function $f_k$ is randomly selected, and we generate
$
Y_i = k_1 f_k(\mathbf{X}_i) + \beta^\top \mathbf{X}_i + \epsilon_i + c_0,
$
where $k_1$ controls the intensity of nonlinear dependency. More detailed setting can be found in appendix \ref{app:simu1}.

In the CRT experiments, we perform $K = 100$ resampling iterations 
for conditional randomization testing. 
In linear case, we benchmark against classical methods including OLS-CUSUM and MOSUM \citep{ploberger1992cusum, MOSUM}. Unless otherwise stated, we adopt the authors' default hyperparameters provided by the original implementations; For the MOSUM, since it is not design for condition CPD, we applying the univariate detector to the residual process. Concretely, we fit a regularized linear regression of $Y$ on $X$ , obtain residuals $\hat{\varepsilon}_i = Y_i - \widehat{\mathbb{E}}[Y_i \mid X_i]$, and run MOSUM on $\{\hat{\varepsilon}_i\}$. 
In non-linear case, we benchmark against the kernel-based methods for detecting changes in conditional expectations and distributions, referred to as KCE and KCD, respectively \citep{nie2022detection}.

\paragraph{Simulation Result}
In the linear change-point setting, our results, as presented in Figure~\ref{fig:Type2-power-p-val}, indicate that the misdetection probability of all methods gradually increases as $k_s$ decreases. When $k_s$ is relatively large---implying a more pronounced structural difference before and after the change point---our method maintains excellent discriminative sensitivity even when $k_s < 0.5$, achieving a markedly lower false detection rate than the other two baselines. As $k_s$ continues to decrease, the false detection rate of our model increases slightly but remains at a consistently low level, demonstrating clear superiority over the competing approaches.

\begin{figure}[H]
    \centering
    \includegraphics[width=1\textwidth]{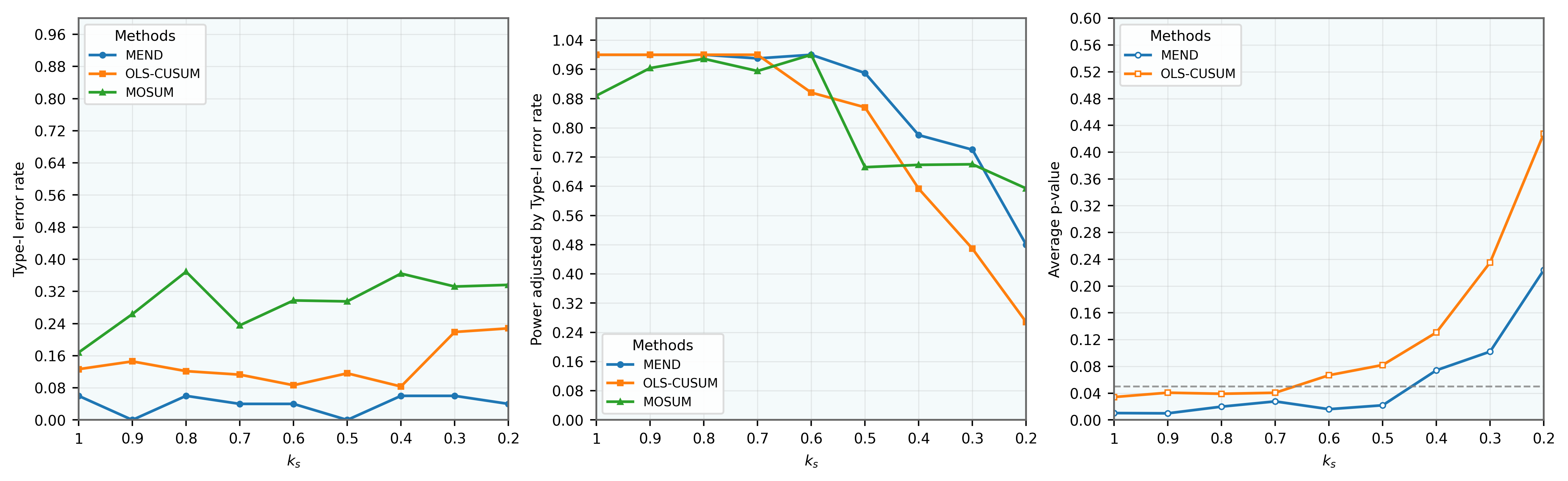}
    \caption{\small{Comparison of Type-$\text{II}$ error rate and average p-value rate across different methods(linear)}}
    \label{fig:Type2-power-p-val}
\end{figure}

In the nonlinear change-point setting, figure~\ref{fig:Type2-power-p-val-non} illustrates the Type-II error rates at significance levels $\alpha = 0.05$ and $\alpha = 0.10$, as well as the average $p$-values across different signal strengths $k_s$.

As the signal weakens (i.e., $k_s$ decreases), all methods exhibit an expected increase in Type-II error. 
However, the MEND relatively maintains the lowest error rates in most $k_s$ both when $\alpha = 0.05$ and $\alpha = 0.10$ . 
Concerning $p$-values ,MEND and KCE both yield low $p$-values (approximately 0.02--0.04) across all $k_s$, 
indicating strong detection sensitivity even under weak linear signals, while KCD yields substantially inflated $p$-values (often exceeding 0.1), failing to identify significant distributional shifts.

Overall, these results demonstrate that the proposed MEND framework achieves the best trade-off between sensitivity and reliability, maintaining stable Type-II error control and low $p$-values even under weak linear signal regimes, 
whereas both kernel-based competitors exhibit pronounced degradation in performance.


\begin{figure}[H]
    \centering
    \includegraphics[width=1\textwidth]{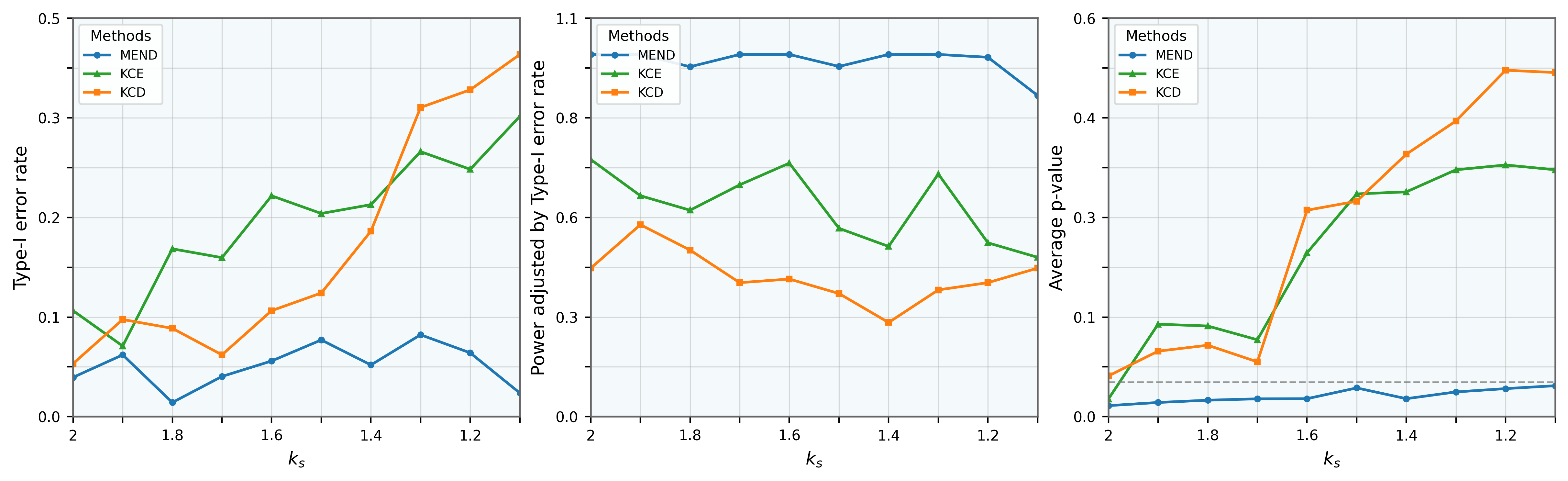}
    \caption{\small{Comparison of Type-$\text{II}$ error rate and average p-value rate across different methods(non-linear)}}
    \label{fig:Type2-power-p-val-non}
\end{figure}

We also conducted simulation data experiments and comparisons on the simplified models \textbf{MEND-mean} and \textbf{MEND-repr} proposed in the previous text. Details can be found in the appendix \ref{app:simu2}.

\section{Real Data Application}\label{sec:application}

\paragraph{Beijing Air Quality Data.}
In this section, we apply the proposed methods to two real-world datasets. The first dataset is the Beijing Air Quality Dataset (2010--2014), originally analyzed in \citet{BeijingPollute}. The data record continuous hourly climate and air pollution measurements from \textit{January 1, 2010} to \textit{December 31, 2014}, including PM2.5 concentration, temperature, atmospheric pressure, wind direction, wind speed, and precipitation conditions.

The original study investigated the response relationship between PM2.5 and environmental variables, aiming to establish a stable pollution classification standard. Building upon this, we analyze whether the conditional dependency of PM2.5 on other meteorological factors (e.g., wind direction, speed, and precipitation) exhibits structural changes over time---potentially associated with major environmental policies, seasonal transitions, or shifts in residential behavior.

According to the referenced study, two major environmental transition events are of particular interest: (i) the annual winter heating season, which substantially increases PM2.5 levels and alters the impact of environmental variables; and (ii) the APEC Summit (Nov 8--11, 2014), during which a series of stringent pollution-control policies were implemented (e.g., \textit{Beijing Air Quality Assurance Plan for APEC 2014}), significantly reducing PM2.5 concentration.


Under a nominal significance level of \( \alpha = 0.05 \), our model identify statistically significant change points. 
displays the estimated test statistic \( D_t \), where four major peaks---corresponding to potential change points---are observed at \textit{Oct 11, 2010}, \textit{Jan 20, 2012}, \textit{Jan 24, 2013}, and \textit{Feb 27, 2014}. These dates align closely with the onset or end of each winter heating period, consistent with the findings in the original study.
When examining the top ten detected points, we observe that the algorithm sensitively captures the seasonal transition intervals including 2010.10.11--2011.02.24, 2011.10.10--2012.01.20, 2013.01.24--2013.03.18 and 2013.10.29--2014.02.27. The official Beijing heating season spans \textit{Nov 15 to Mar 15}, with minor yearly adjustments based on temperature, showing strong agreement with our detection results.

\begin{figure}[H]
    \centering
    \includegraphics[width=0.8\textwidth]{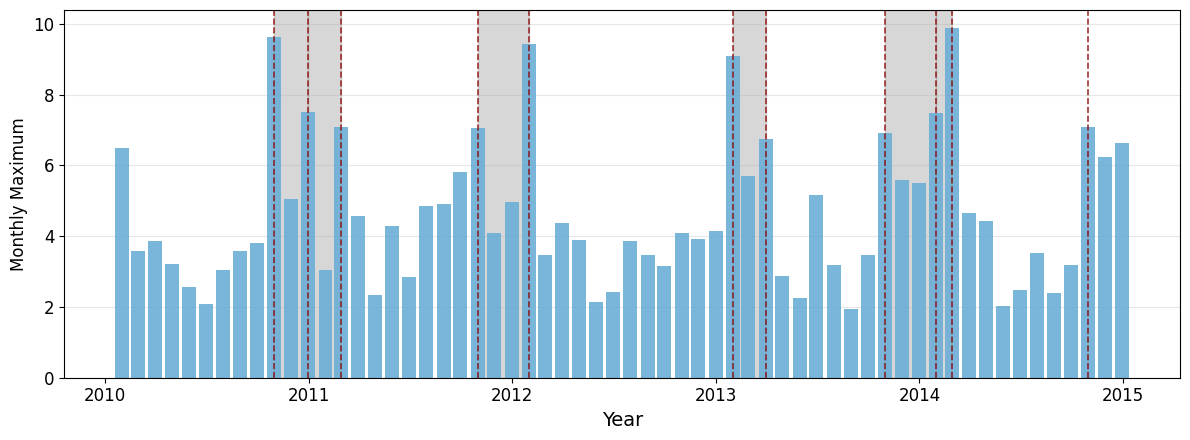}
    \caption{\small{Application in Beijing Air Quality Data}}
    \label{fig:Beijing-Air-Quality-Data}
\end{figure}

Moreover, our model identifies another distinct change point on \textit{Oct 26, 2014}, corresponding precisely to the APEC Summit period (Nov 8--11, 2014), during which strict pollution-control measures---including traffic restrictions and factory shutdowns---were enforced. This further validates the effectiveness and interpretability of our method in detecting policy-driven environmental changes.


\paragraph{Duke Electronic Health Records Data.}
We also apply the developed methods to real-world electronic health records (EHR) data drawn from the Duke University Health System via the Duke Clinical Research Datamart, including all emergency department (ED) visits in $2022$. The outcome of interest is a binary variable indicating whether a visit results in inpatient admission to the hospital. The predictors include demographic variables (age, sex, race), vital signs (pulse rate, systolic blood pressure, diastolic blood pressure, oxygen saturation, temperature, respiration rate), and selected comorbidities (local tumor, metastatic tumor, diabetes with complications, diabetes without complications, and renal disease), defined according to ICD codes.

EDs operate in a dynamic clinical environment, and there can be a seasonal model shift in these data—for instance, higher ED visits leading to hospitalization during peak influenza season \citep{Menec2003}. Thus, our goal is to determine whether the proposed MEND algorithm can detect temporal shifts of this nature. We randomly selected a subsample of $n=1,000$ observations from the full cohort. To avoid possible multiple change points within a single year, we further restricted our analysis to data collected from April through December $2022$.

We first conduct a pseudo-simulation based on the covariates in real data and assume that there is no change point. Type-I error for MEND is 0.029. Type-I error for all the other methods are smaller than 0.05. This indicates a good quality of modeling $R\mid \mathbf{X}$. Details of the pseudo-simulation can be found in Appendix \ref{sec:pseudo}.

We then applied the methods to the original data set. Table \ref{tab2} shows the estimated change point and corresponding $p$-values of all methods. Our proposed method MEND give small $p$-value with estimated change point 11, indicating that there is a change in November. This align with the flu season which potentially lead to the data shift. Baseline DP gives the same point estimation as the proposed method, and other baseline methods have large $p$-value that not able to identify significant change point. These results illustrate that our method shows higher statistical power than existing methods;


\begin{table}[!htb]
  \setlength\tabcolsep{2pt}
  \renewcommand\arraystretch{1}
  \centering
  \small
    \caption{\small{Estimated change points and $p$-values of different methods in our real-world example}}  \label{tab2}
    \begin{tabular}{c|cccccc}
      \hline
      &MEND & KCE &KCD &DP & OLS-CUSUM \\
      \hline\hline
Estimated $\tau$ (month) &11 &6 &8 &11 &11 \\
 $p$-value &0.029 & 0.742 &0.6 & - &0.55\\
   \hline
    \end{tabular}
\end{table}

\section{Conclusion}

In this paper, we propose a novel method for detecting change points in conditional distributions, 
built upon the Conditional Randomization Test (CRT) framework and enhanced with a neural network–based distillation function. 
The proposed approach is applicable to high-dimensional and complex dependence settings, 
while effectively alleviating the computational burden of CRT caused by high-dimensional feature spaces. 
Moreover, it avoids strong parametric assumptions on the conditional distribution of the response, 
achieving robust and flexible inference across diverse data-generating mechanisms. 
We further introduce two simplified variants of the model designed for low-dimensional or single-change scenarios, 
which can efficiently complete the detection task with reduced computational complexity. 
Both theoretical analysis and numerical studies demonstrate the statistical validity and empirical effectiveness of the proposed framework.

Nevertheless, similar to other neural-network–based methods, 
the selection of hyperparameters remains a nontrivial and potentially sensitive process, 
which constitutes one of the current limitations of our approach. 
Additionally, due to the use of deep distillation representations, 
the model lacks interpretability and cannot provide an explicit functional form of the dependency between \( X \) and \( Y \).

Future research may extend the current framework toward distinguishing between 
\textit{structural} and \textit{numerical} change points, 
thereby identifying whether the dependency structure or the magnitude of the relationship has shifted. 
Such extension would require stronger supervision or structural priors 
and may further enhance the interpretability and diagnostic power of the proposed method.

\newpage
\bibliographystyle{plainnat}
\bibliography{references}

\newpage
\appendix
\onecolumn
\section{Stabilized EM Training and IRM Regularization}
\label{app:alg1}

The EM algorithm we apply to train $\phi$ can be represented as follow:

\paragraph{E-step.}
Given the current parameters $(\phi, \omega_k, \sigma_k^2)$, we compute the responsibility of each sample:
\begin{align*}
P[E_i = k \mid (\mathbf{x}_i, y_i)]
&= \frac{P[(\mathbf{x}_i, y_i) \mid E_i = k] \, P(E_i = k)}
{\sum_{k'} P[(\mathbf{x}_i, y_i) \mid E_i = k'] \, P(E_i = k')} \\[6pt]
&= \frac{P[(\mathbf{x}_i, y_i) \mid E_i = k]}
{\sum_{k'} P[(\mathbf{x}_i, y_i) \mid E_i = k']} \\[6pt]
&= 
\frac{
\mathcal{N}\!\big(y_i - \omega_k^\top \phi(\mathbf{x}_i), \sigma_k^2 \big)
}{
\sum_{k'} 
\mathcal{N}\!\big(y_i - \omega_{k'}^\top \phi(\mathbf{x}_i), \sigma_{k'}^2 \big)
} \\[6pt]
&= p_{ik}.
\end{align*}

\paragraph{M-step.}
Using the soft assignments $\{p_{ik}\}$, we update both the encoder and decoders 
by minimizing the weighted reconstruction loss:
\[
L = L_{\text{MSE}} + \lambda_1 L_{\text{reg}}
= \sum_{i=1}^{n}\sum_{k=1}^{K}
p_{ik}\, \big(y_i - \omega_k^\top \phi(\mathbf{x}_i)\big)^2
+ \lambda_1 \big(\|\phi\|_2^2 + \|\omega\|_2^2\big),
\]
where $\lambda_1$ is a regularization coefficient controlling overfitting.
Both $\phi$ and $\{\omega_k\}$ are updated via gradient descent.
Mini-batch training can be adopted for large datasets, 
and additional normalization or regularization can be applied to stabilize convergence.

The variance parameters $\sigma_k^2$ admit closed-form updates:
\[
Q(\sigma_k^2) = \sum_i p_{ik} \log P(y_i \mid \mathbf{x}_i),
\quad
\frac{\partial Q}{\partial \sigma_k^2} = 0
\;\;\Rightarrow\;\;
\sigma_k^2 = 
\frac{\sum_i p_{ik} (y_i - \omega_k^\top \phi(\mathbf{x}_i))^2}
{\sum_i p_{ik}}.
\]

During the EM training of $\phi$, we introduce the idea of \textbf{Invariant Risk Minimization (IRM)} to further stabilize the learning process.  
The overall loss function is modified as
\begin{equation}
L = L_{\mathrm{MSE}} + \lambda_1 L_{\mathrm{reg}} + \lambda_2 L_{\mathrm{env}},
\end{equation}
where
\begin{equation}
L_{\mathrm{env}} = \left\| \nabla_{\bar{\omega}} L_{\mathrm{MSE}}(\bar{\omega} \cdot \phi) \right\|^2, 
\quad 
\bar{\omega} = \sum_{k} \bar{p}_{k} \omega_k, 
\quad 
\bar{p}_{k} = \sum_i p_{ik}.
\end{equation}

The first term $L_{\mathrm{MSE}}$ is the standard mean squared error loss for prediction,  
the second term $L_{\mathrm{reg}}$ is a simple $\ell_2$ regularization term,  
and the third term $L_{\mathrm{env}}$ serves as a stabilization component inspired by the IRM principle.

Specifically, for each data sample $(x_i, y_i)$, it is probabilistically assigned to multiple linear heads,  
and the weighted summation of their outputs is taken as the final prediction.  
Under this stochastic assignment, $L_{\mathrm{env}}$ measures the squared $\ell_2$-norm of the gradient of $L_{\mathrm{MSE}}$  
with respect to this ``mixture of linear heads.''  
This formulation is analogous to the IRM-v2 objective:
\begin{equation}
\sum_e \left\| \nabla_{\omega} R^e(\omega \cdot \phi) \right\|^2,
\end{equation}
with the only difference being that the environment index $e$ is replaced here by a softer probabilistic representation,  
where $\bar{p}_k$ acts as weighting coefficients across different linear heads.

\section{Simplified Model for the Single-Change-Point Case}
\label{app:Simplified}

\subsection{MEND-mean} 
When it is known a prior that the data contains at most one change point 
and the dependency between $Y$ and $\mathbf{X}$ is approximately linear, 
it is unnecessary to employ complex neural architectures. 
Instead, we adopt a simplified conditional expectation model for both modeling and distillation. 
This corresponds to using a single-layer linear network as the distillation function $\phi$.

The simplest form of the distillation function, denoted by $\phi_0(\mathbf{X})$, 
is constructed from the conditional means estimated via the EM algorithm. 
Specifically, the EM step outputs two estimated conditional means, 
$\hat{m}_0(\mathbf{X})$ and $\hat{m}_1(\mathbf{X})$, 
which correspond to the pre-change and post-change regimes, respectively.
The distilled representation is defined as

\[
\phi_0(\mathbf{X}) = \{\hat{m}_0(\mathbf{X}), \hat{m}_1(\mathbf{X})\},
\]
\[
m_0(X_i) = \mathbb{E}[Y_i \mid X_i,\; i < \tau], \qquad   m_1(X_i) = E[Y_i|X_i, i\ge \tau ]
\]

which concatenates the two regime-specific conditional expectations of $Y$.

For each candidate change point $\tau$, 
we partition the data into subsets $\{R_i \le \tau\}$ and $\{R_i > \tau\}$,
and model the observations using the conditional mixture
\[
Y_i \sim a \, \hat{m}_0(\mathbf{X}_i) + (1 - a)\, \hat{m}_1(\mathbf{X}_i),
\]
where $a$ represents a regime-specific mixing coefficient. 
For each side of $\tau$, we estimate $a$ by solving the following regularized least-squares problems:
\begin{equation}
\label{eq:coefficients}
\begin{aligned}
\hat{a}_{\tau-} &= \arg\min_a 
\sum_{R_i \le R_\tau} \left( Y_i - a \hat{m}_0(\mathbf{X}_i) - (1 - a) \hat{m}_1(\mathbf{X}_i) \right)^2 
+ \lambda (a - 0.5)^2, \\[4pt]
\hat{a}_{\tau+} &= \arg\min_a 
\sum_{R_i > R_\tau} \left( Y_i - a \hat{m}_0(\mathbf{X}_i) - (1 - a) \hat{m}_1(\mathbf{X}_i) \right)^2 
+ \lambda (a - 0.5)^2,
\end{aligned}
\end{equation}
where $\lambda$ is a small regularization parameter introduced to stabilize the coefficient estimation,
especially when one side of the split contains relatively few samples.

The divergence between the two regimes is measured by the difference of their estimated mixing coefficients.
The resulting test statistic is defined as
\[
S(\mathcal{D}) = 
\max_{\tau \in \{1, \dots, T-1\}} 
\left| \hat{a}_{\tau-} - \hat{a}_{\tau+} \right|.
\]
The $d_0$ approach relies solely on the conditional means of $Y$,
thereby achieving high computational efficiency 
while remaining effective in scenarios where the conditional mean captures 
the essential dependency structure between $Y$ and $\mathbf{X}$.

\subsection{MEND-LAD (repr)}

Moreover, we can artificially introduce additional variables into the linear-layer function to strengthen the specific statistical properties of interest. If we wish to emphasize the oscillatory behavior of the function, we can include the change in variance before and after the change point as part of the fitting process. If our focus lies on the influence of noise perturbations, we may incorporate the learning and estimation of noise coefficients within the distillation function $\phi$. When the relationship exhibits strong linearity, we can even directly output certain selected components of $X$ to enhance the linear characteristics.
In this section, under the simple single–change–point scenario, we introduce an additional method:
Specifically, let $B_0$ and $B_1$ denote the sets of indices corresponding to the top $k$ features (e.g., based on coefficient magnitude) from the fitted models $\hat{m}_0$ and $\hat{m}_1$, respectively. Then, we define 
 $
 \phi(\mathbf{X}) = \left\{ \hat{m}_0(\mathbf{X}),\ \hat{m}_1(\mathbf{X}),\ \mathbf{X}_{:B_0 \cup B_1} \right\}. 
 $
Our framework is designed to be flexible and can incorporate any such representation that facilitates sensitive and robust inference.

\section{Proofs}
\label{appendix:proof}
\begin{proposition}
\label{prop:crt}    
Assume we have sequential data $D = \{(Y_{i}, \mathbf{X}_i, R_i)\}_{i=1}^{n}$ Given $R_i$ with each observation $(Y_{i}, \mathbf{X}_i)$ follows the generative model: 
\begin{equation}
Y,\mathbf{X} \mid R \sim  p_{\mathbf{X}|R} (\mathbf{X})   p_{Y|\mathbf{X}}(Y),
\end{equation}
Now define $D^{*} = \{(Y_i, \mathbf{X}_i, R_i^{*})\}_{i=1}^{n}$, where each $R_i^{*} \sim p_{R|\mathbf{X}_i}(t)$, then any test statistic $S(D)$ obeys:
\begin{equation}
S(D^*)|(Y,\mathbf{X}) \stackrel{d}{=} S(D)|(Y,\mathbf{X})
\end{equation}
\end{proposition}

\begin{proof}[Proof of Proposition \ref{prop:crt}]
Under the model, given \( R \), the covariates \( \mathbf{X} \) are generated first from \( p_{\mathbf{X}|R}(\mathbf{X}) \), and then the response \( Y \) is generated from \( p_{Y|\mathbf{X}}(Y) \) independently of \( R \). This implies the conditional independence:
\begin{equation}
Y \perp R \mid \mathbf{X}
\end{equation}
To prove the claim, it suffices to show that \( R \) and \( R^* \) have the same distribution conditionally on \( (Y, \mathbf{X}) \). This follows from

\[
R^* \mid (Y, \mathbf{X}) \stackrel{d}{=} R \mid X \stackrel{d}{=} R \mid (Y, \mathbf{X}) .
\]

The first equality comes from the definition of \(R^* \) while the second follows from the conditional independence of \( Y \) and \( R \).
and therefore any statistic computed from \( (Y, \mathbf{X}, R) \) has the same distribution as one computed from \( (Y, \mathbf{X}, R^*) \), conditionally on \( (Y, \mathbf{X}) \)
\end{proof}

\subsection{Proof of Theorem \ref{theo:valid_test}}
Assume the null hypothesis
\[
H_0: p_{Y|\mathbf{X},R}(y) = p_{Y|\mathbf{X}}(y), \quad \text{for all } R \in \{1, \ldots, T\},
\]
holds. If we use the test statistic \( S(D) \) and the resampled data sets \(D^{(k)}\) defined in Algorithm~\ref{alg:crt_modified}, then by Proposition~\ref{prop:crt}, we have:
\begin{equation}
S(D^{(k)}) \mid (Y, \mathbf{X}) \stackrel{d}{=} S(D) \mid (Y, \mathbf{X}), \quad \text{for all } k = 1, \ldots, K.
\end{equation}

Therefore, the $p$-value is defined as 
\[
p\text{-value} = \frac{1}{K+1} \left( 1 + \sum_{k=1}^K \mathbb{I}\big(S(D^{(k)}) \geq S(D)\big) \right)
\]
is valid under \( H_0 \). Specifically,
\[
P(p\text{-value} \leq \alpha) = P\left(S(D^{(k)}) > S(D) \text{ for less than } \alpha (K+1) \text{ values of } k\right) \leq \alpha.
\]
Hence, under \( H_0 \),
\[
\lim_{n \to \infty} \mathbb{P}\left(\textnormal{Algorithm~\ref{alg:crt_modified} rejects}\right) \leq \alpha.
\]

\section{Numerical Simulation on MEND-mean and MEND-repr}
\label{app:simu2}
\subsection{Simuation settings}

We conduct several simulation test on our simplified models.
Specifically, we consider a total of \(T=10\) time points and set the change point position at \(\tau=7\). At each time point $t$, $n_t=100$ observations are generated, with a total sample size of \(n=1000\). The \(X\) generation is the same as Linear Change-Point Scenario above.

For the generation of $Y$, we consider the following three scenarios, where the first two scenarios are low-dimensional nonlinear models and the third scenario is high-dimensional linear model.

{\bf Scenario 1.} Set
$
    Y_{i}= \mathbf{X}_i^\top \alpha^*  + \delta_1 x_{i1}^2 +\epsilon_{i},
$
where $\boldsymbol{\alpha}^*=(0.5,-0.5, 0.5, 0.5, -0.5, 0, \ldots,0)\in \mathbb{R}^{20}$, and the  strength of the nonlinear perturbation $\delta_1$ is varied over $\{0.1, 0.2, 0.3, 0.4, 0.5\}$ to evaluate the ability of each method to control the Type-I error inflation under increasing  nonlinear effects. The noise term \(\epsilon_i\) is independently sampled from \(N(0,1)\), which remains consistent across all scenarios. 
Furthermore, we conduct an additional experiment by fixing $\delta_1=0$ and varying $N$ to assess the quality of model-X in our proposed methods.

{\bf Scenario 2.} Set $Y_{i} = \boldsymbol{S}(\mathbf{X}_i) \alpha^* + \epsilon_{i}$, if $ R_i \le \tau$, and  $Y_{i} = \boldsymbol{S}(\mathbf{X}_i) (\alpha^* + \delta_2\beta) + \epsilon_{i}$, if $ R_i > \tau$,

where $\mathbf{X}_i \in \mathbb{R}^5$ is drawn from $N(g(R_i), \boldsymbol{\Sigma})$ with $g(R_i) = \mathbf{0}$ for $R_i \leq 5$ and $g(R_i) = 0.2 R_i \mathbf{1}_5$ otherwise. The transformation $\boldsymbol{S}(\mathbf{X}_i) = (\sin(x_{i1}), x_{i2}^3, x_{i3}^2, x_{i4}, x_{i5}^2)$ induces nonlinear effects. Here, we set $\boldsymbol{\alpha}^* = (0.5, -0.5, 0.5, 0.5, -0.5)$ and  $\boldsymbol{\beta}=(0.05, 0.05, 0.05, 0.05, 0.05)$. The signal strength of structural change $\delta_2$ is set as $0$ and $3$ to examine the ability to control Type-I error in the presence of covariate shift without structural change, and to assess the power to detect true structural change under a nonlinear model, respectively.

{\bf Scenario 3.} Set $ Y_{i}=\mathbf{X}_i^\top \alpha^* +\epsilon_{i}$, if $ R_i \le \tau$, and $ Y_{i}=  \mathbf{X}_i^\top (\alpha^*+\delta_3\beta) +\epsilon_{i}$, if $ R_i > \tau$, where  $\boldsymbol{\alpha}^*=(0.5,-0.5, 0.5, 0.5, -0.5, 0, \ldots,0)\in \mathbb{R}^{100}$, $ \boldsymbol{\beta}=(0.05, 0.05, 0.05, 0.05, 0.05, 0, \ldots,0)\in \mathbb{R}^{100}$,

and the signal strength of the change in the regression coefficients $\delta_3$ is varied over $\{0,1,2,3,4,5\}$ to evaluate the power and localization performance under increasing structural change strength in high-dimensional settings.

Similarly, we compare MEND-mean and MEND-repr with three benchmark approaches:
(1) OLS-CUSUM for detecting the changes in low-dimensional parametric models ,
(2) the dynamic programming method  (DP) for localizing changes in high-dimensional parametric models  \citep{rinaldo2021localizing}, and
(3) KCE and KCD. To evaluate the performance of the change point detection and localization, we investigate the Type-I error rate and power of testing $H_0$ against $H_1$, the localization accuracy rate, and the run time. The localization accuracy rate is defined as the proportion of replications in which the estimated change point exactly coincides with the true change point. The significance level is set to $\alpha=0.05$. 

Following Nie and Nicolae (2022), the p-values in KCE and KCD is determined by 500 bootstraps.
Since KCE and KCD are not directly applicable in high-dimensional settings, we perform variable
screening using the SCAD penalty (Fan and Li, 2001) before applying these methods to ensure
efficient and accurate analysis in the high-dimensional setting. For the kernel bandwidth h, it is
tuned among $ S_h = \{0.001, 0.1, 0.2, 0.3, 0.4, 0.5, 1, 3, 5, 10, 15\}$.
The bandwidths h for KCD and KCE are set as follows

\begin{table}
  \setlength\tabcolsep{2pt}
  \renewcommand\arraystretch{1.2}
  \centering
    \caption{The selected bandwidth $h$ for KCD and KCE}  \label{tab3}
    \begin{tabular}{c|c|c|c}
      \hline
   $h$ &Scenario 1 & Scenario 2 &Scenario 3 \\
     \hline\hline
   h-KCD &0.4   &0.5 &0.5 \\
   \hline
   h-KCE &0.2   &0.2 &0.1 \\
   \hline
    \end{tabular}
\end{table}

\paragraph{Simulation Result} In Figure~\ref{fig:Type1-power-accuracy}, we compare the Type-I error control in Scenario 1 and the detection power and localization accuracy in Scenario 3 of our methods with the benchmark methods. Since some methods fail to control the Type-I error rate, we report the power adjusted by subtracting the Type-I error rate in Figure~\ref{fig1} to ensure a fair comparison. The Type-I error rate performance for $\delta_3=0$ in Scenario 3 is presented in Table \ref{tab4}.  Since DP is not designed for change-point testing, we exclude it from Figures~\ref{fig3}-\ref{fig1} and Table \ref{tab4}.  Furthermore, given that OLS-CUSUM is designed for low-dimensional model,  we only report the performance of other methods in Scenario 3. As shown in Figures~\ref{fig:Type1-power-accuracy} and Table~ \ref{tab4}, our proposed methods consistently outperform the benchmark methods across all evaluation metrics and scenarios. To be specific, Figure~\ref{fig3} demonstrates that the proposed method MEND maintains more stable Type-I error control as the signal strength of the nonlinear term $\delta_1$ increases. In contrast, the Type-I error rate of OLS-CUSUM inflates with the strength of the  nonlinear effect, while both KCE and KCD exhibit severe Type-I error inflation, even when the nonlinear effect is small. For instance, when $\delta_1=0.3$, the benchmark methods KCD, KCE, and OLS-CUSUM exhibit Type-I error rates of 0.48, 0.4, and 0.1, respectively, while the proposed MEND-mean method achieves a much lower rate of 0.046.


\begin{figure}[t]
    \centering

    \begin{subfigure}[b]{0.3\textwidth}
        \centering
        \includegraphics[width=\textwidth,height=3.8cm]{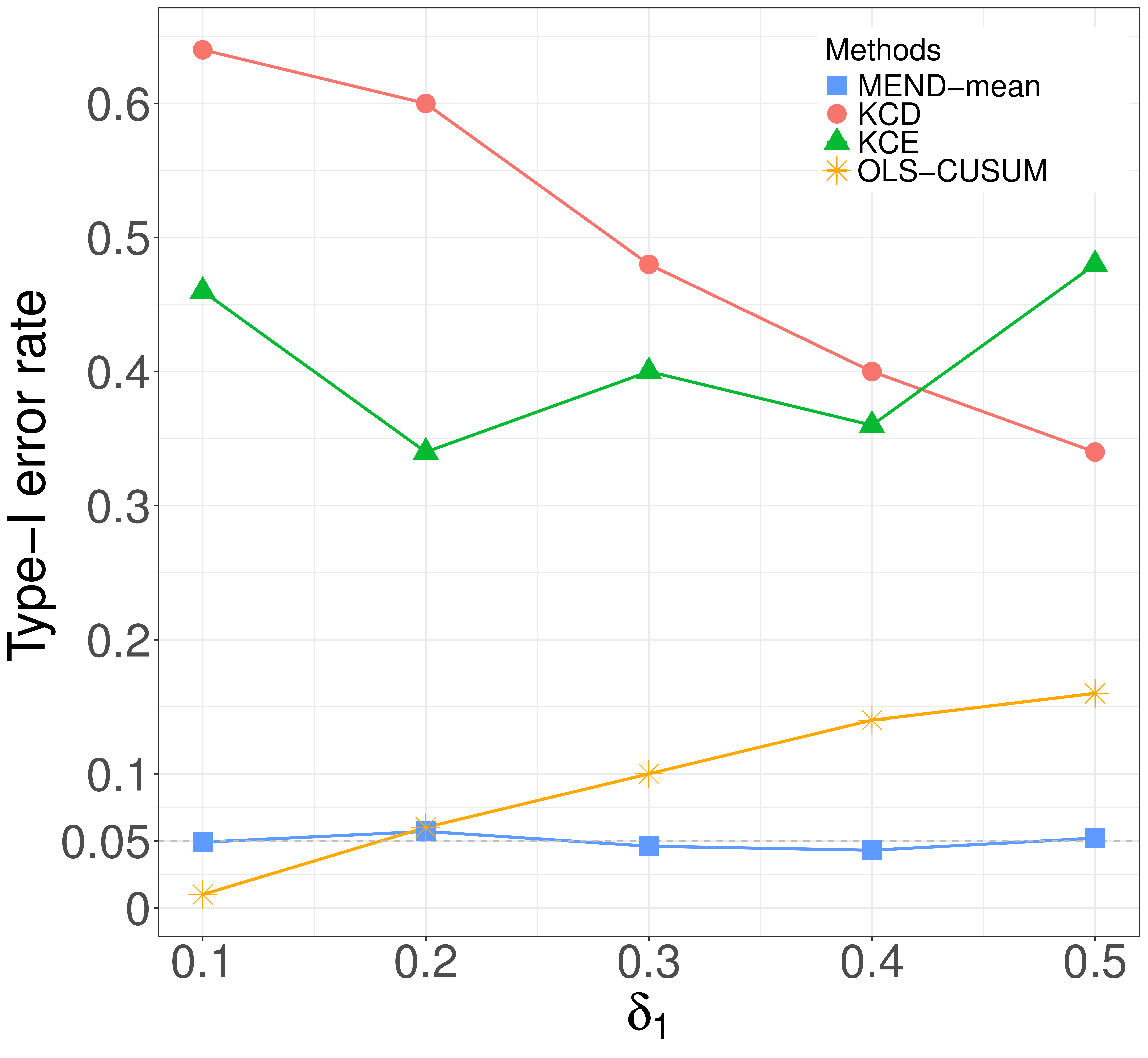}
        \caption{\small Type-I error rate in Scenario 1 with increasing nonlinear effect size $\delta_1$}
        \label{fig3}
    \end{subfigure}
    \hfill
    \begin{subfigure}[b]{0.3\textwidth}
        \centering
        \includegraphics[width=\textwidth,height=3.8cm]{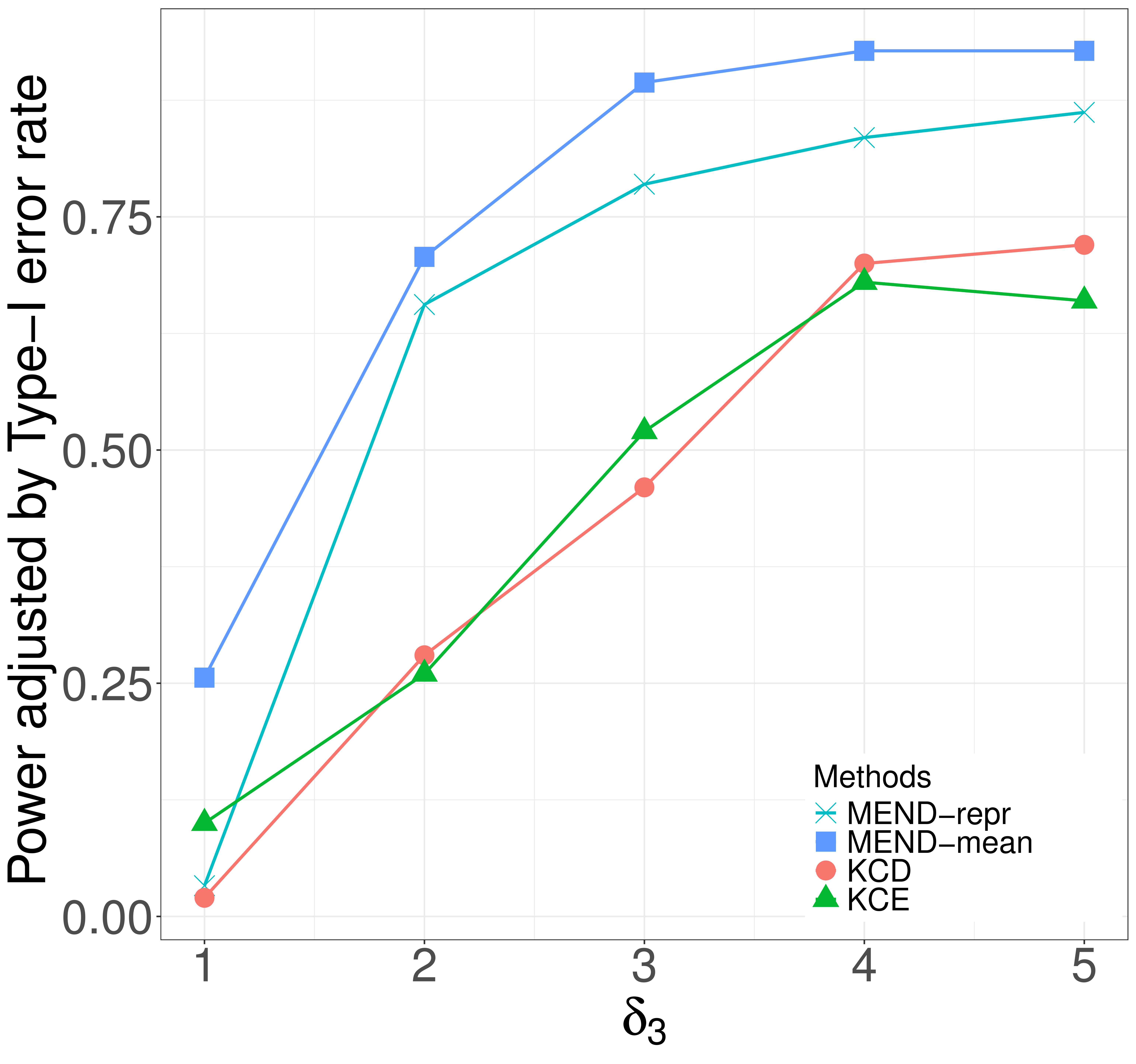}
        \caption{\small Power adjusted by Type-I error rate in Scenario 3 with increasing $\delta_3$}
        \label{fig1}
    \end{subfigure}
    \hfill
    \begin{subfigure}[b]{0.3\textwidth}
        \centering
        \includegraphics[width=\textwidth,height=3.8cm]{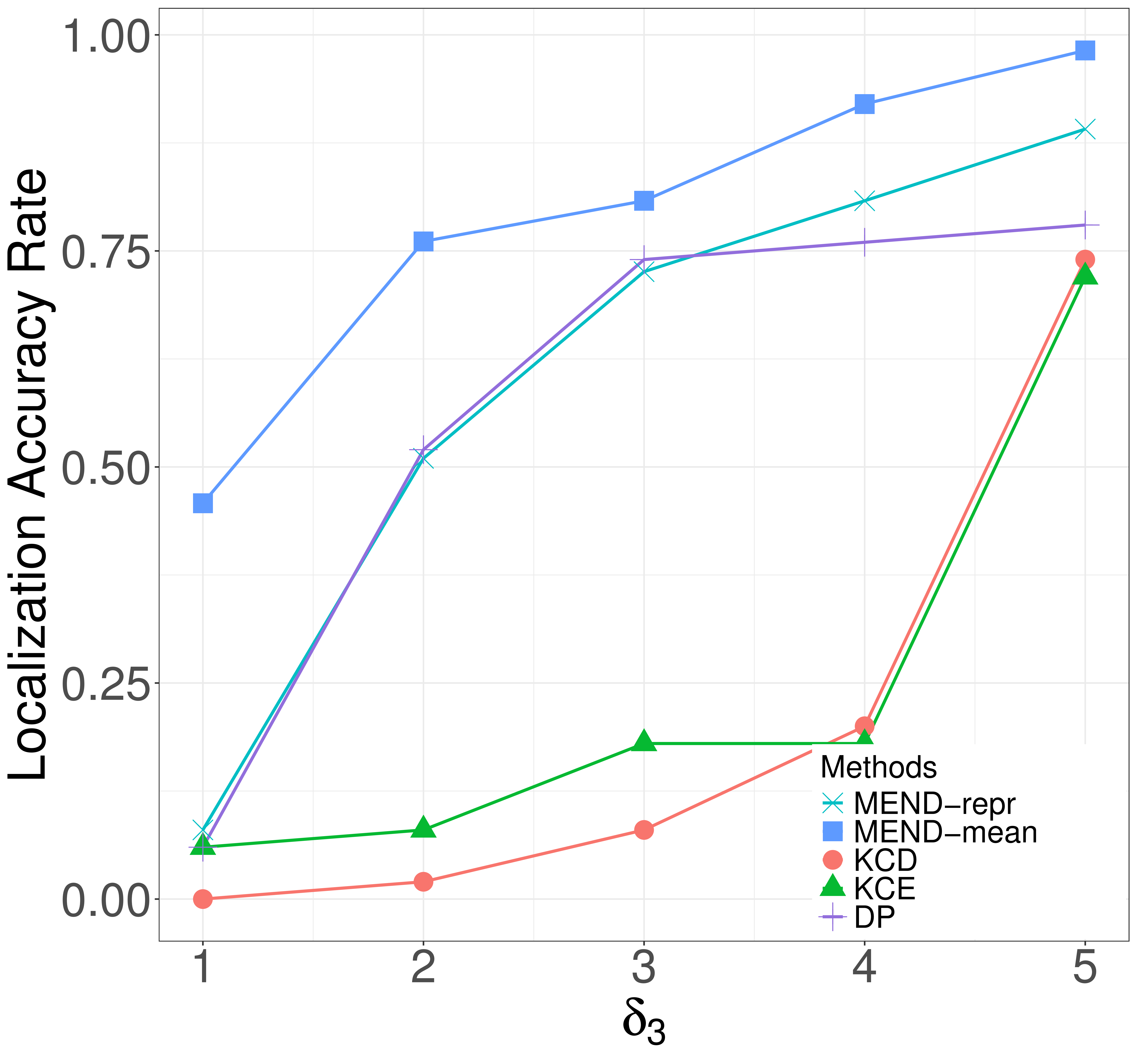}
        \caption{\small Localization accuracy rate in Scenario 3 with increasing  $\delta_3$}
        \label{fig2}
    \end{subfigure}

    \caption{\small Comparison of Type-I error rate, power and accuracy rate across different scenarios and methods (linear)}
    \label{fig:Type1-power-accuracy}
\end{figure}

\begin{table}[H]
  \setlength\tabcolsep{2pt}
  \renewcommand\arraystretch{1}
  \centering
  \small
    \caption{\small{Averaged Type-I error rates of different methods over 500 replications in Scenario 3}}  \label{tab4}
    \begin{tabular}{c|cccc}
    
      \hline
   $\delta_3$      &$\text{MEND-mean}$   & MEND-repr    &KCD  &KCE                            \\
     \hline\hline
     0     &0.054 &  0.057 &  0.28 &0.24 \\
   \hline
    \end{tabular}
\end{table} 
In view of Figures~\ref{fig1}-\ref{fig2}, both the power and the localization accuracy rate of all methods increase as the signal strength of the structural change $\delta_3$ in the regression coefficient grows. However, our proposed methods MEND-mean and  MEND-repr have uniformly higher power and localization accuracy rate than the benchmark methods. For example, when $\delta_3=3$, KCD and KCE achieve powers of $0.74$ and $0.76$ and localization accuracy rates of $0.08$ and $0.18$, respectively, while MEND-mean and  MEND-repr  attain substantially higher powers of $0.959$ and $0.842$ and localization accuracy rates of $0.808$ and $0.726$. Similarly, Table~\ref{tab4}  shows that the proposed MEND-mean and  MEND-repr methods achieve significantly lower Type-I error rates of $0.054$ and $0.057$ compared to $0.28$ for KCD and $0.24$ for KCE, indicating much better error control under the null setting.

\begin{table}[H]
  \setlength\tabcolsep{2pt}
  \renewcommand\arraystretch{1}
  \centering
  \small
    \caption{\small{Averaged results of different methods over 500 replications in Scenario 2} } \label{tab1}
    \begin{tabular}{cc|ccccc}
    
      \hline
   $\delta_2$ &     &MEND-mean   &KCD  &KCE    &DP  &OLS-CUSUM                          \\

     \hline\hline
     0    &Type-I error rate &0.054 &  0.63 &0.47    & -  &0.51\\
     \hline
   3             & Power        &0.65  &0.76  &0.39   & -   &1  \\
                & Localization Accuracy Rate &0.65 &0.07    &0    &0.01  &0.58      \\
   \hline
    \end{tabular}
\end{table}

Table \ref{tab1} compares the Type-I error rate, detection power, and the localization accuracy of different methods in Scenario 2. It shows that our proposed method MEND-mean achieves the best overall performance. It maintains the lowest Type-I error rate of $0.054$ when the signal strength of structural change $\delta_2=0$, while all other methods suffer from severe Type-I error inflation, with rates exceeding $0.45$. As $\delta_2$ increases to $3$, the proposed method attains a desirable balance between the power and the localization accuracy rate, whereas the benchmarks exhibit either low power or poor localization accuracy.

\section{Details on Numerical Studies}
\label{app:simu1}
In the Nonlinear Change-Point Scenario in section 4.1, the nonlinear function set $F$ is defined as:
\[
\begin{aligned}
f_1(\mathbf{X}) &= 0.6|X_0|^2 + 2\sqrt{|X_2|},\\
f_2(\mathbf{X}) &= 0.3|X_0|X_1 + 1.5\sqrt{|X_2|},\\
f_3(\mathbf{X}) &= 0.4|X_0| + \sqrt{|X_2|} + e^{X_0/2},\\
f_4(\mathbf{X}) &= 0.2|X_4| + 1.5\sqrt{|X_2|} + e^{(X_0 - |X_4|)/2},\\
f_5(\mathbf{X}) &= 0.3|X_4| +  e^{(X_0 - |X_4|/2)/2}.
\end{aligned}
\]

The shared encoder (distillation head) $\phi$ is implemented as a 
four-layer fully connected neural network, 
with ReLU activations applied to the first two layers to enhance nonlinearity. 
For the fused lasso training, we fix the regularization parameters as:
$
\lambda_{\mathrm{fuse}} = 1, 
\lambda_{l1} = 1, 
\lambda_{\mathrm{ridge}} = 0.1.
$

\section{Details on Beijing PM2.5 dataset test}
We first preprocess the data, remove missing values, and apply one-hot encoding to categorical features. We treat the PM2.5 concentration as the response variable \( Y \), and all remaining features as predictors \( X \). The resulting dataset contains \( 41{,}757 \) samples with an initial feature dimension of \( p = 10 \). Data collected within the same day are aggregated and treated as belonging to the same time index \( T \). 

The shared encoder (distillation head) $\phi$ is implemented as a 
four-layer fully connected neural network, 
with ReLU activations applied to the first two layers to enhance nonlinearity. 
For the fused lasso training, we fix the regularization parameters as:
$
\lambda_{\mathrm{fuse}} = 1, 
\lambda_{l1} = 0.5, 
\lambda_{\mathrm{ridge}} = 0.5.
$

\section{Details on Pseudo-simulation}
\label{sec:pseudo}
For the pseudo simulation, the covariates $\mathbf{X}_i$ are taken from the real data. We assume there is no change point, meaning that the data across all months follow the same underlying model. Pseudo outcomes are generated from a binomial distribution:

$$Y_\text{pseudo,$i$}\sim \text{Binomial}\left(\frac{\exp(\hat\eta^{T}\mathbf{X}_i)}{1+\exp(\hat\eta^{T}\mathbf{X}_i)}\right)$$

where $\hat\eta$ is a fixed coefficient vector applied uniformly across all months. The Type-I error rate is estimated by calculating the $p$-values for each method using $\mathbf{X}_i$ and $Y_{\text{pseudo},i}$, and repeating this process 500 times.








\end{document}